\documentclass[leqno,a4paper]{article}
\usepackage{amssymb,amsmath,amsthm,verbatim}
\usepackage[latin1]{inputenc}
\usepackage[T1]{fontenc}
\usepackage{url}
\usepackage{graphicx}

\title{\textbf{A group law on the projective plane\\with applications in Public Key Cryptography}}
\author{\textsc{R. Durán Díaz$^1$, V. Gayoso Martínez$^2$}
\and
\textsc{L. Hernández Encinas$^2$, and J. Muñoz Masqué$^2$}
\bigskip \\
$^1$ Departamento de Automática, Universidad de Alcalá,\\
E-28871 Alcalá de Henares, Spain\\
E-mail: \texttt{raul.duran@uah.es}\\
$^2$ Instituto de Tecnologías Físicas y de la Información (ITEFI)\\
Consejo Superior de Investigaciones Científicas (CSIC),\\
E-28006 Madrid, Spain\\
E-mails: \texttt{\{victor.gayoso, luis, jaime\}@iec.csic.es}
}

\date{}

\newtheorem{theorem}{Theorem}[section]
\newtheorem{proposition}[theorem]{Proposition}
\newtheorem{lemma}[theorem]{Lemma}
\newtheorem{corollary}[theorem]{Corollary}

\theoremstyle{remark}
\newtheorem{remark}[theorem]{Remark}

\begin{document}

\maketitle

\begin{abstract}
\noindent
We present a new group law defined on a subset of the projective plane
$\mathbb{F}P^2$ over an arbitrary field $\mathbb{F}$, which lends itself to
applications in Public Key Cryptography, in particular to a Diffie-Hellman-like
key agreement protocol. We analyze the computational difficulty of solving the
mathematical problem underlying the proposed Abelian group law and we prove that
the security of our proposal is equivalent to the discrete logarithm problem
in the multiplicative group of the cubic extension of the finite field
considered. Finally, we present a variant of the proposed group law but over the
ring $\mathbb{Z}/pq\mathbb{Z}$, and explain how the security becomes enhanced,
though at the cost of a longer key length.
\end{abstract}

\medskip



\section{Introduction}

The main contribution of this paper is to propose a new group law, defined on
the complement of a projective cubic plane curve, prove its
properties, and consider the possibility of using it as a building block for
cryptographic applications in the field of Public Key Cryptography (PKC).

The paper is organized as follows: Section~\ref{S:Grouplaw} presents the group
law and its main characteristics and properties. In particular, we define the
mathematical problem associated to the considered group law, and we give the
explicit formulas to compute the group operation of any two elements of the
group. These formulas, which involve coefficients from the base field, are
applicable to any pair of elements of the group with no exception whatsoever,
which is advantageous in view of possible cryptographic applications.

As an application of the defined group law to PKC, a cryptographic protocol, in
particular, a Diffie-Hellman-like key agreement protocol, is defined in
section~\ref{S:Crypto}. We also analyze the computational difficulty of solving
the mathematical problem underlying the defined group law, and we prove that the
hardness of our problem is equivalent to that of the discrete logarithm problem
on the multiplicative group of the cubic extension of the finite field
considered.

In section~\ref{S:Robust} we consider an entirely analogous system, but
shifting the general base field to the ring $\mathbb{Z}/pq\mathbb{Z}$. We
make it clear that this last proposal enhances the security of the system,
since it now depends not only on DLP but also on the factorization problem,
though at the price of doubling the key length.

Last section is devoted to the conclusions.

\section{The group law defined}\label{S:Grouplaw}

Let $\mathbb{F}$ be a field and let us consider a linear endomorphism
$A\colon V\to V$ of the vector space $V=\mathbb{F}^3$.
We define the polynomial $Q(\mathbf{x})=\det (x_1I+x_2A+x_3A^2)$, where
$\mathbf{x}=(x_1,x_2,x_3)\in V$. The polynomial $Q$ is homogeneous of degree $3$,
and does not depend on $A$, but only on the characteristic polynomial
$\chi(X)$ of $A$.

A new group law is proposed $\oplus\colon V\times V\to V$. Let the multiplicative
group $\mathbb{F}^{\ast }$ act on $V$ by the diagonal action, i.e.,
$\lambda \cdot(x_{1},x_{2},x_{3}) =(\lambda x_{1},\lambda x_{2},\lambda x_{3})$,
and let denote by $\mathbb{F}P^{2}$ the projective plane, namely
$\mathbb{F}P^{2}=(V\setminus\{ (0,0,0)\} )/\mathbb{F}^{\ast}$.
Then the proposed group law induces an Abelian group law
on $\mathbb{F}P^{2}\setminus Q^{-1}(0)$.

If the characteristic polynomial $\chi (X)$ is irreducible in $\mathbb{F}[X]$,
then $Q^{-1}(0)=\emptyset $. In this case, the group law extends to the whole
set $\mathbb{F}P^2$; moreover, if the base field is a finite field
$\mathbb{F}_q$, with characteristic different from $2$ or $3$, then the group
$\mathbb{G}=(\mathbb{F} _qP^2,\oplus )$ is proved to be cyclic.

The latter property permits us to apply the notion of \emph{discrete logarithm}
to the group $\mathbb{G}$. If we fix a generator $g\in \mathbb{F}_qP^2$, then any
element $h$ of the group is the addition of $g$ with itself a finite number of
times, say $n$, so that $h=g\oplus g \oplus \overset{(n}\cdots\oplus g = [n]g$.
The number $n$ is the logarithm of $h$ to the base $g$.

Given any element $h\in\mathbb{G}$, and a generator $g$ of the group, the
\emph{discrete logarithm problem} (DLP), consists in finding the smallest
integer $n$, such that $h=[n]g$. In this work, we prove that the DLP over
$\mathbb{G}$ with a proper choice of the generator is equivalent to the DLP
over the multiplicative group $(\mathbb{F}_{q^3})^\ast $.

Popular current cryptosystems are based on the discrete logarithm problem over
different groups, such as the group of invertible elements in a finite field, or
the group of points of an elliptic curve with the addition of points as group
operation. Our proposal could fit perfectly well in the same niche.

As is the case for analogous public key protocols, the users of the present
proposal agree to a single base field $\mathbb{F}_q$ but each one of them is
allowed to select at will any (irreducible) polynomial
\[
\begin{array}
[c]{rll}
\chi (X)= & X^3-c_1X^2-c_2X-c_3,
& c_1,c_2,c_3\in \mathbb{F}_q.
\end{array}
\]
The public system parameters include the base field $\mathbb{F}_q$,
coefficients $c_1, c_2, c_3\in \mathbb{F}_q$, and the generator $g$.

Next we prove that the polynomial $Q$ does not depend on $A$, but only on the
characteristic polynomial $\chi(X)$ of $A$.
\begin{lemma}
\label{lemma1}
Let $\mathbb{F}$ be a field and let $V$ be the vector space
$\mathbb{F}^3$. If $A\colon V\to V$ is a linear map such that
the endomorphisms $I,A,A^2$ are linearly independent,
then the homogeneous cubic polynomial
$Q(\mathbf{x})=\det (x_1I+x_2A+x_3A^2)$ does not depend on the matrix
$A$ but only on the coefficients $c_1,c_2,c_3$ of its
characteristic polynomial $\chi (X)=X^3-c_1X^2-c_2X-c_3$.
\end{lemma}
\begin{proof}
Let $\mathbb{\bar{F}}$ be the algebraic closure of $\mathbb{F}$.
As the endomorphisms $I,A,A^2$ are linearly independent,
the annihilator polynomial of $A$ coincides with $\chi (X)$
by virtue of the Cayley-Hamilton theorem. Hence there exists
a basis of $\mathbb{\bar{F}}^3$ such that the matrix
of $A$ in this basis equals one of the following three matrices:

\begin{equation}\label{matrix1}
M_{1}  =  \left(
\begin{array}{ccc}
\alpha _1 & 0 & 0 \\
0 & \alpha _2 & 0 \\
0 & 0 & \alpha _3
\end{array}
\right),\,
M_2 = \left(
\begin{array}{ccc}
\alpha _1 & 0 & 0 \\
0 & \alpha _2 & 0 \\
0 & 1 & \alpha _2
\end{array}
\right),\,
M_3 = \left(
\begin{array}{ccc}
\alpha _1 & 0 & 0 \\
1 & \alpha _1 & 0 \\
0 & 1 & \alpha _1
\end{array}
\right),
\end{equation}
and from a simple calculation we obtain
\begin{eqnarray}
Q(\mathbf{x}) & = &\det (x_1I+x_2M_i+x_3(M_i)^2)
\label{P(x)} \\
& = &-c_2x_1(x_2)^2
+\left[ (c_2)^2-2(c_1c_3)\right]
x_1(x_3)^2+c_1(x_1)^2x_2
\notag \\
&& +\left[ (c_1)^2+2c_2\right]
(x_1)^2x_3-(c_2c_3)x_2(x_3)^2+(c_1c_3)(x_2)^2x_3
\notag \\
&& -\left( c_1c_2+3c_3\right)
x_1x_2x_3+(x_1)^3+c_3(x_2)^3+(c_3)^2(x_3)^3,
\notag
\end{eqnarray}
for every $i=1,2,3$.
\end{proof}
\begin{theorem}
\label{th1}
Every linear map $A\colon V\to V$ such that the
endomorphisms $I,A,A^2$ are linearly independent,
induces a law of composition
\begin{equation*}
\begin{array}{c}
\oplus \colon V\times V\to V, \\
(\mathbf{x},\mathbf{y})\mapsto \mathbf{z}
=\mathbf{x}\oplus \mathbf{y},
\end{array}
\end{equation*}
by the following formula:
\begin{equation}
\begin{array}{rl}
z_1I+z_2A+z_3A^2= & \left( x_1I+x_2A+x_3A^2\right)
\left(
y_1I+y_2A+y_3A^2\right) ,
\end{array}
\label{law}
\end{equation}
where $\mathbf{x}=(x_1,x_2,x_3)$, $\mathbf{y}=(y_1,y_2,y_3)$,
$\mathbf{z}=(z_1,z_2,z_3)$.

Moreover, the set of elements $\mathbf{x}\in V$ such that
$\mathbf{x}\oplus \mathbf{y}=(0,0,0)$ for some element
$\mathbf{y}$ in $V\setminus \{ (0,0,0)\} $ coincides with
the set $Q^{-1}(0) $, and $\oplus $ induces a group law
\begin{equation*}
\oplus \colon (\mathbb{F}^3\setminus Q^{-1}(0))
\times (\mathbb{F}^3
\setminus Q^{-1}(0))\to (\mathbb{F}^3\setminus Q^{-1}(0)).
\end{equation*}
If $C$ denotes the projective cubic curve defined
by $Q(\mathbf{x})=0$, then the group law $\oplus $ also induces
a group law
\begin{equation*}
\oplus \colon (\mathbb{F}P^2\setminus C)
\times (\mathbb{F}P^2\setminus C)
\to \mathbb{F}P^2\setminus C.
\end{equation*}
\end{theorem}
\label{th2}
\begin{proof}
As $A^3=c_1A^2+c_2A+c_3I$, and
\begin{align*}
A^2\cdot A^2& =A\cdot A^3 \\
& =\left( c_1c_3\right) I+\left( c_1c_2+c_3\right) A
+\left[ (c_1)^2+c_2\right] A^2,
\end{align*}
from the formula in \eqref{law} it follows:
\begin{equation}
\begin{array}{rl}
z_1= & x_1y_1+c_3\left( x_2y_3+x_3y_2\right)
+\left(
c_1c_3\right) x_3y_3, \\
z_2= & x_1y_2+x_2y_1+c_2\left( x_2y_3+x_3y_2\right)
+\left( c_1c_2+c_3\right) x_3y_3, \\
z_3= & x_2y_2+x_1y_3+x_3y_1+c_1\left(
x_2y_3+x_3y_2\right) +\left( (c_1)^2+c_2\right) x_3y_3.
\end{array}
\label{z's}
\end{equation}
In matrix notation, these formulas can equivalently
be written as
\begin{equation*}
\left(
\begin{array}{c}
z_1 \\
z_2 \\
z_3
\end{array}
\right) =\left(
\begin{array}{ccc}
x_1 & c_3x_3 & c_1c_3x_3+c_3x_2 \\
x_2 & x_1+c_2x_3 & c_2x_2+c_3x_3+c_1c_2x_3 \\
x_3 & x_2+c_1x_3 & x_1+(c_1)^2x_3+c_1x_2+c_2x_3
\end{array}
\right) \left(
\begin{array}{c}
y_1 \\
y_2 \\
y_3
\end{array}
\right) ,
\end{equation*}
and as a simple computation shows, the determinant
of the linear system above is equal to $Q(\mathbf{x})$,
where $Q$ is defined by the formula \eqref{P(x)}.
Hence $\mathbf{x}\oplus \mathbf{y}=(0,0,0)$, for some
$\mathbf{y}$ in $V\setminus \{ (0,0,0)\} $, if and only if
$Q(\mathbf{x})=0$.

The commutativity of $\oplus $ is a direct consequence
of the invariance of the formula \eqref{z's} under
the substitutions $x_i\mapsto y_i$, $y_i\mapsto x_i$,
$1\leq i\leq3$.

Moreover, the formula \eqref{law} can also be written as follows:
\begin{multline*}
\left( \mathbf{x}\oplus \mathbf{y}\right) _1I
+\left( \mathbf{x}\oplus \mathbf{y}\right) _2A
+\left( \mathbf{x}\oplus \mathbf{y}\right) _3A^2 =\\
\left( x_1I+x_2A+x_3A^2\right) \left( y_1I+y_2A+y_3A^2\right) .
\end{multline*}
From the associativity of the composition law of endomorphisms we deduce
\[
\begin{array}[c]{l}
\left( \mathbf{x}\oplus (\mathbf{y}\oplus \mathbf{z})\right) _1I
+\left( \mathbf{x}\oplus (\mathbf{y}\oplus \mathbf{z})\right) _2
A+\left( \mathbf{x}\oplus (\mathbf{y}\oplus \mathbf{z})\right) _3A^2
\medskip\\
=\left( x_1I+x_2A+x_3A^2\right) \cdot \left( \left( y_1I+y_2A+y_3A^2\right)
\cdot \left( z_1I+z_2A+z_3A^2\right) \right)
\medskip\\
=\left( \left( x_1I+x_2A+x_3A^2\right) \cdot
\left( y_1I+y_2A+y_3A^2\right) \right) \cdot
\left( z_1I+z_2A+z_3A^2\right)
\medskip\\
=\left( (\mathbf{x}\oplus \mathbf{y})\oplus \mathbf{z}\right) _1
I+\left( (\mathbf{x}\oplus \mathbf{y})\oplus \mathbf{z}\right) _2A
+\left( (\mathbf{x}\oplus \mathbf{y})\oplus \mathbf{z}\right) _3A^2.
\end{array}
\]
Hence $\mathbf{x}\oplus (\mathbf{y}\oplus \mathbf{z})
=(\mathbf{x}\oplus \mathbf{y})\oplus \mathbf{z}$,
$\forall \mathbf{x},\mathbf{y},\mathbf{z}\in V$.

From \eqref{z's} it follows that the unit
element is the point $(1,0,0)$, which does not
belong to $Q^{-1}(0)$ since $Q(1,0,0)=1$.

By taking determinants in the equation \eqref{law}
we obtain
\begin{equation*}
\begin{array}{lll}
Q(\mathbf{x}\oplus \mathbf{y})= & Q(\mathbf{x})Q(\mathbf{y}),
& \forall \mathbf{x},\mathbf{y}\in V.
\end{array}
\end{equation*}
Therefore the opposite element $\mathbf{y}$ of $\mathbf{x}$ exists
and it is given by the following formulas:
\begin{equation*}
\begin{array}{rl}
y_1=\! & \!\! \tfrac{c_1x_1x_2+\left[ (c_1)^2+2c_2\right]
x_1x_3-\left( c_3+c_1c_2\right)
x_2x_3+(x_1)^2-c_2(x_2)^2+\left[ (c_2)^2-c_1c_3\right]
(x_3)^2}{Q(\mathbf{x})},\smallskip \\
y_2=\! & \!\! - \tfrac{x_1x_2+(c_1)^2x_2x_3+c_1(x_2)^2
-\left( c_1c_2+c_3\right) (x_3)^2}{Q(\mathbf{x})},
\smallskip \\
y_3=\! & \!\!
\tfrac{-x_1x_3+c_1x_2x_3+(x_2)^2-c_2(x_3)^2}{Q(\mathbf{x})}.
\end{array}
\end{equation*}

Finally, if $\mathbf{x}$, $\mathbf{y}$ are replaced by
$\lambda \mathbf{x}$, $\mu \mathbf{y}$,
respectively, with $\lambda ,\mu \in \mathbb{F}^\ast $,
then $\mathbf{z}$ transforms into $\lambda \mu \mathbf{z}$,
thus proving that the group law projects onto
$\mathbb{F}P^2\setminus C$.
\end{proof}
\begin{remark}
Note that the equations in~\eqref{z's}, allowing one to compute
the $\oplus $ group operation in terms of the coefficients
in the ground field, are applicable to any element of the group,
with no exception at all.
\end{remark}
\begin{remark}
If $\mathbf{v}_1=(1,0,0)$, $\mathbf{v}_2=(0,1,0)$, $\mathbf{v}_3=(0,0,1)$,
then from \eqref{P(x)} we obtain $Q(\mathbf{v}_2)=c_3$, $Q(\mathbf{v}_3)=(c_3)^2$.
Hence $\mathbf{v}_2$ and $\mathbf{v}_3$ belong to
$\mathbb{F}^3\setminus Q^{-1}(0)$ if and only if $c_3\neq 0$,
i.e., when $A$ is invertible.
\end{remark}
\subsection{The basic cubic}
\begin{proposition}
\label{prop1}
Let $\chi (X)=X^3-c_1X^2-c_2X-c_3\in \mathbb{F}[X]$
be the polynomial introduced in \emph{Lemma~\ref{lemma1}}
and let $\alpha =X\bmod \chi $.
If $N\colon \mathbb{F}[\alpha ]\to \mathbb{F}$
is the norm of the extension
$\mathbb{F}[\alpha ]$ of $\mathbb{F}$,
then a point $\beta =\beta _0+\beta _1\alpha +\beta _2\alpha ^2$
belongs to the cubic curve $C$ defined in
\emph{Theorem~\ref{th1}} if and only if $N(\beta )=0$.
In particular, if $\chi $ is irreducible in $\mathbb{F}[X]$,
then $C$ has no point in $\mathbb{F}P^2$.

Moreover, the polynomial $\chi $ is irreducible
 in $\mathbb{F}[X]$ if and only if the cubic $C$ is irreducible.
\end{proposition}
\begin{proof}
Every $\beta \in \mathbb{F}[\alpha ]$ induces
an $\mathbb{F}$-linear endomorphism
$E_\beta \colon \mathbb{F}[\alpha ]\to \mathbb{F}[\alpha ]$
given by
$E_\beta (\xi )=\beta \cdot \xi$,
$\forall \xi \in \mathbb{F}[\alpha ]$,
and from the very definition of the norm we have
$N(\beta )=\det E_\beta $. As a computation shows,
we obtain $N(\beta )=Q(\beta _0,\beta _1,\beta _2)$,
thus proving the first part of the statement.
Moreover, $\chi $ is irreducible if and only if
$\mathbb{F}[\alpha ]$ is a field and then the norm is injective,
thus proving the second part of the statement.

Finally, if $\chi $ factors in $\mathbb{F}[X]$, say
$X^3-c_1X^2-c_2X-c_3=(X-h)(X^2+kX+l)$, with
$h,k,l\in \mathbb{F}$, then we have
\[
Q(\mathbf{x})=[(x_1)^2+(k^2-2l)x_1x_3+l(x_2)^2-klx_2x_3+l^2(x_3)^2-kx_1x_2]
[x_1+hx_2+h^2x_3].
\]

Conversely, if $\chi $ is irreducible in $\mathbb{F}[X]$,
then according to Proposition~\ref{prop1}, the only solution
to the cubic equation $Q(\mathbf{x})=0$ is $\mathbf{x}=\mathbf{0}$.
Hence $Q$ must be irreducible, as a reducible cubic admits
non-trivial solutions in the ground field.
\end{proof}
\begin{corollary}
\label{cor1}
If the characteristic polynomial $\chi $ of $A$ is irreducible
in $\mathbb{F}[X]$, then there is no linear transformation
$(\lambda _{ij})_{i,j=1}^3\in GL(\mathbb{F},3)$ reducing
the polynomial $Q$ defined in \eqref{P(x)} to Weierstrass form.
\end{corollary}
\begin{proof}
Replacing $x_j$ by $X_j=\sum _{i=1}^3\lambda _{ij}x_i$,
$1\leq j\leq 3$, in \eqref{P(x)} we obtain a cubic $\bar{Q}$,
which is in Weierstrass form (see \cite[\S 2.1]{Menezes})
if and only if the coefficients $a$, $b$, and $c$
of the terms $(x_3)^3$, $(x_1)^2x_2$, and
$x_1(x_2)^2$, respectively, vanish. As a computation shows,
we have
$a=\bar{Q}(\lambda _{31},\lambda _{32},\lambda _{33})$,
and we can conclude by applying Proposition~\ref{prop1}.
\end{proof}
\subsection{Cyclicity}
\begin{theorem}
\label{th4}
If $\mathbb{F}_q$ is a finite field of characteristic
different from $2$ or $3$ and the polynomial
$\chi (X)=X^3-c_1X^2-c_2X-c_3$ introduced
in \emph{Lemma~\ref{lemma1}} is irreducible in
$\mathbb{F}_q[X]$, then the group $\mathbb{G}=(\mathbb{F}_qP^2,
\oplus)$ is cyclic.
\end{theorem}
\begin{proof}
Since $\operatorname{char}\mathbb{F}_q\neq 2,3$,
the polynomial $\chi $ is separable and in its splitting field
$\mathbb{F}_q^\prime $ we have
$\chi (X) =(X-\alpha _1)(X-\alpha _2)(X-\alpha _3)$,
the roots $\alpha _1$, $\alpha _2$, $\alpha _3$ being pairwise
distinct, and in a certain basis of
$\mathbb{F}_q^\prime \otimes _{\mathbb{F}_q}V$
the matrix of $A$ is given by the formula \eqref{matrix1}.
As the Galois group $G(\mathbb{F}_q^\prime /\mathbb{F}_q)$
acts transitively on the roots of $\chi $, there exist
two automorphisms such that $\sigma _2(\alpha _1)=\alpha _2$
and $\sigma _3(\alpha _1)=\alpha _3$. If
$\beta =\beta _1+\beta _2\alpha _1+\beta _3(\alpha _1)^2$,
$\beta _i\in \mathbb{F}_q$, $1\leq i\leq 3$, is an element in
$\mathbb{F}_q[\alpha _1]\cong \mathbb{F}_{q^3}$, then
for every positive integer $n$ we have
\[
\left( \beta _1I+\beta _2A+\beta _3A^2\right) ^n
=\left(
\begin{array}{ccc}
\beta ^n & 0 & 0 \\
0 & \sigma _2\left( \beta ^n\right)  & 0 \\
0 & 0 & \sigma _3\left( \beta ^n\right)
\end{array}
\right) .
\]
Consequently, if $\beta $ is a generator
of the multiplicative group $(\mathbb{F}_{q^3})^\ast $,
then the vector $(\beta _1,\beta _2,\beta _3)$ generates
the group
$((\mathbb{F}_q)^3\setminus \{ (0,0,0)\} ,\oplus)$
and its corresponding projective point
$[\beta _1,\beta _2,\beta _3]
=(\beta _1,\beta _2,\beta _3)\bmod \mathbb{F}_q^\ast $
generates the group $\mathbb{G}$, with
$\mathbb{F}_qP^2
=\left( (\mathbb{F}_q)^3\setminus \{ (0,0,0)\} \right)/\mathbb{F}_q^\ast $.
\end{proof}

\begin{remark}\label{rem1}
It is important to keep in mind that the implication in Theorem~\ref{th4}
works only in the way in which it is worded. If one selects a generator of the
group $\mathbb{G}$, it will in general be a
generator of only a subgroup of the whole $(\mathbb{F}_{q^3})^\ast $ group.
Consequently, when choosing a generator for $\mathbb{G}$, it is
convenient to pick it from the set of generators in $(\mathbb{F}_{q^3})^\ast $
and, \emph{after that}, project it onto $\mathbb{F}_qP^2$.
\end{remark}

\begin{remark}
As the order of the group $\mathbb{G}=(\mathbb{F}_qP^2,\oplus)$ is $q^2+q+1$,
the statement of Theorem~\ref{th4} means that there exists an element
$\beta \in \mathbb{G}$ of order $q^2+q+1$. According to the proof
of Theorem~\ref{th4} this is equivalent to saying that the matrix $A$ in
\eqref{matrix1} is of order $q^2+q+1$ in the linear group $GL(\mathbb{F}_q,3)$.
A classical result (see \cite[Theorem, p.\ 379]{Singer}) states
that such a collineation always exists, but we need a direct proof
of this fact to be able to apply it below in section~\ref{cryptanalysis};
also see \cite[Proposition 2.1]{GHK}.
\end{remark}

\begin{remark}
When the polynomial $\chi $ is reducible, experimental tests carried out in the
prime field $\mathbb{F}_p$ show that the projective cubic curve $C$ defined as
$Q(\mathbf{x}) = 0$ has a number of points from the set $\{ p+2, 2p+1, 3p, p+1\} $ only.

Since the projective space $\mathbb{F} _pP^2$ has a total of $p^2+p+1$ points,
the group $(\mathbb{F} _pP^2\setminus C,\oplus)$ is left,
respectively, with $\{p^2-1, p^2-p, (p-1)^2, p^2\}$ points.

If the number of points of $C$ is either $p+2$ or $2p+1$, then the group
$(\mathbb{F} _pP^2\setminus C,\oplus)$ is still cyclic, and has the
expected number of generators, namely, either $\varphi(p^2-1)$ or $\varphi(p^2-p)$,
respectively, where $\varphi$ is Euler's totient function.

However none of the other two possibilities give rise to a cyclic group.
Rather, for the case where $C$ has $3p$ points, there appears a number of cyclic
groups, whose cardinalities are the divisors of $p-1$;
it is important to remark that the total
number of points left for the group is precisely $(p-1)^2$. Thus, the group
$(\mathbb{F} _pP^2\setminus C,\oplus)$ can be decomposed as a direct sum of a
number of cyclic groups such that the product of their cardinalities is
$(p-1)^2$.

As for the case when $C$ has $p+1$ points, the group $(\mathbb{F} _pP^2
\setminus C,\oplus)$ is not cyclic either and can be decomposed as a direct sum
of $2$ cyclic groups with $p$ points each. Remark that now the total number of
points left for the group is $p^2$, so again the numbers of points of the cyclic
groups of this case match the divisors of $p$.
\end{remark}
\section{A cryptographic protocol}\label{S:Crypto}
First of all, we establish the computational security of the
mathematical problem defined over the cyclic group considered.
Later on, as an example of cryptographical protocol,
we present a Diffie-Hellman-like key agreement protocol.

\subsection{Equivalence of DLP in $\mathbb{G}$ and
$(\mathbb{F}_{q^3})^\ast $}
\label{cryptanalysis}
\begin{proposition}
\label{prop2}
Let $\mathbb{F}_q$ be a finite field
of characteristic $\neq 2$ or $3$.
Assume the polynomial $\chi (X)=X^3-c_1X^2-c_2X-c_3$
in \emph{Lemma~\ref{lemma1}} is irreducible in $\mathbb{F}_q[X]$,
and let $\alpha \in \mathbb{F}_{q^3}$ be a root
of $\chi $.

If $(\gamma _1,\gamma _2, \gamma _3)$
is a generator of the group
$((\mathbb{F}_q)^3\setminus \{(0,0,0)\},\oplus)$
and $(\beta _1,\beta _2,\beta _3)$ belongs to this group,
then $n\in \mathbb{N}$ is a solution to the equation
\begin{equation*}
\left( \beta _1,\beta _2,\beta _3\right)
=\left( \gamma _1,\gamma _2,\gamma _3\right)
\oplus \overset{(n}{\ldots }\oplus
\left( \gamma _1,\gamma _2,\gamma _3\right) ,
\end{equation*}
if and only if $n$ is a solution to the equation
$\beta =\gamma ^n$ in the multiplicative group
$(\mathbb{F}_{q^3})^\ast $, where
$\beta =\beta _1 +\beta _2\alpha +\beta _3\alpha ^2$,
and
$\gamma =\gamma _1 +\gamma _2\alpha +\gamma _3\alpha ^2$.

Therefore, the DLP in the group $((\mathbb{F}_q)^3\setminus\{(0,0,0)\},\oplus)$
is equivalent to the DLP in $(\mathbb{F}_{q^3})^\ast $.
\end{proposition}
\begin{proof}
Letting $\alpha =\alpha _1$, the statement follows
from the matrix formula in the proof of Theorem~\ref{th4}
taking the very definition of the group law $\oplus $
by the formula \eqref{law} into account.
\end{proof}

In the present case, Proposition~\ref{prop2} states the ``equivalence'' because
the reduction of problems (see, for example, \cite[p.\ 5]{KM},
\cite[Ch.\ 8]{Papadimitriou}) works both ways, namely, DLP in the group
$((\mathbb{F}_q)^3\setminus\{(0,0,0)\},\oplus)$
reduces to the DLP in $(\mathbb{F}_{q^3})^\ast $ and the other way around.
Hence, Proposition~\ref{prop2} proves that the use of the group
$\mathbb{G}=(\mathbb{F}_qP^2,\oplus )$  is safe for standard implementations
in PKC (e.g., see \cite[\S 1.6]{Menezes}), since the security
it provides is equivalent to that of DLP in $(\mathbb{F}_{q^3})^\ast $,
as long as the caveat stated in Remark~\ref{rem1} is taken into account.

In terms of cryptanalysis, in principle logarithms in $\mathbb{G}$ can be
computed using ``generic'' algorithms, i.e., those that assume no particular
structure in (or extra knowledge of) the group. The most popular ones are
Pohlig-Hellman (which reduces the computation in the whole group to the
computation of the logarithm in all subgroups of prime order of $\mathbb{G}$),
Shank's Baby Step/Giant Step, and Pollard's Rho algorithm. All of them need an
exponential computation time.

However, there exists the so-called index-calculus algorithm, which is much
faster as it is able to compute discrete algorithm in the multiplicative
group of a finite field in subexponential time (see, e.g., \cite{Odlyzko2}).
Since the operations in the proposed group $\mathbb{G} = (\mathbb{F}_qP^2,\oplus
)$ can be efficiently transferred to those in $(\mathbb{F}_{q^3})^\ast$, it
follows that index-calculus algorithm can be applied to the multiplicative group
of the latter. This fact does not render the group operation automatically
useless in the face of possible cryptographic applications, as long as proper
key lengths are utilized.

For general finite fields, such as the proposed one, with a multiplicative group
of size $N$, current state-of-the-art algorithms (including index-calculus)
report computation times of
\begin{equation}\label{LN}
L_{N}(\alpha,c) = \exp\left((c+o(1))(\log N)^\alpha (\log\log N)^{1-\alpha}
                      \right),
\end{equation}
where $\alpha$ and $c$ are parameters in the ranges $0<\alpha<1$ and $c>0$
(sometimes $c$ is omitted and we default to $L_N(\alpha)$). Actually, $\alpha$
drives the transition from an exponential-time algorithm (when $\alpha$
approaches $1$) to a pure polynomial-time algorithm (as $\alpha$ tends to $0$).

The first subexponential algorithms had complexity $L_N(1/2)$ and applied only
to prime fields. Soon $L_N(1/3)$ was achieved for any finite field, with values
for $c$ ranging from $(64/3)^{1/3}$ for fields with high characteristic to
$(128/9)^{1/3}$ for medium characteristic. When dealing with small
characteristic fields, recent research brought down the complexity to $L_N(1/4)$
(\cite{Joux}) and even to quasi-polynomial time (\cite{BGJ}, \cite{GKZ2}). If
the group size is $N=p^n$, and we write $p=L_{p^n}(l_p)$, then the
characteristic is considered ``small'', ``medium-sized'' or ``large''
depending on whether $l_p \leq 1/3$,  $1/3 < l_p < 2/3$, or $l_p \geq 2/3$,
respectively.

In any case, the previous results have been applied in practice and several
cryptanalysis have been successfully carried out (see \cite{AMOR-H},
\cite{KDLPS}), so it seems sensible to avoid using small characteristics and
also extensions of moderate characteristic included in the range threatened by
recent cryptanalytic techniques (\cite{BGJ}, \cite{GKZ2}, \cite{HAKT}).
However these algorithms are heuristic and are proved to work only
for certain particular cases, not difficult to circumvent: for example, if one
has $N=p^n$ it suffices to choose both $p$ and $n$ to be prime in order to
thwart both \cite{BGJ} and \cite{GKZ2}. For a detailed account of history and
current status, see \cite{JOP} (in particular \S 4.2), and \cite{GKZ1}.

Our proposal is to use a group $\mathbb{G}$ of prime order $n=q^2+q+1$, over a
ground field $\mathbb{F}_q$. Using formula \eqref{LN} we can compute how
many elements in $\mathbb{G}$ provide a given security level. Since the
number of elements is roughly the square of the value of $q$, it follows that
$q$ can be represented with only one half of the bits needed for $n$. This has
a direct impact on the computation time of the $\oplus$ operation in
$\mathbb{G}$, since it is performed in $\mathbb{F}_q$ (see equations
\eqref{z's} and cost analysis in subsection \ref{cost}).

\subsection{System set-up and system parameters for a key agreement protocol}
The group $\mathbb{G} = (\mathbb{F}_qP^2,\oplus )$ lends readily itself as a
building block for standard cryptographic applications to be constructed upon
it. One of such applications is a Diffie-Hellman-like key agreement protocol,
which will be described in the following sections.

In the following, we provide the necessary steps to set up the system. Moreover,
the users also need to fix some system parameters.

\bigskip

\noindent\textbf{System set-up}

\medskip

To set up the system, the following steps are in order:
\begin{enumerate}
\item Choose a ground field $\mathbb{F}_q$ with characteristic different from
$2$ or $3$.
\item
Select elements $c_1, c_2, c_3\in \mathbb{F}_q$ such that the polynomial
\[
\chi (X) = X^3-c_1X^2-c_2X-c_3
\]
is irreducible in $\mathbb{F}_q[X]$.
\item
Consider $\mathbb{F}_{q^3}\simeq \mathbb{F}_q[X]/(\chi (X))$. Select
$\alpha \in (\mathbb{F}_{q^3})^\ast $ such that it is a generator of
$(\mathbb{F}_{q^3})^\ast $.
\item
Compute the coordinates of $\alpha $ seen as a vector over
$\mathbb{F}_q$, which will be denoted as
$(\alpha_1, \alpha_2, \alpha_3)\in (\mathbb{F}_q)^3\setminus \{0,0,0\} $.
\item
Under the canonical projection
$\pi \colon (\mathbb{F}_q)^3\setminus\{0,0,0\} \to \mathbb{F}_qP^2$,
compute
$[\beta _1,\beta _2,\beta _3]=\pi (\alpha _1,\alpha _2,\alpha _3)$.
\end{enumerate}

\bigskip

\noindent\textbf{System parameters}

\medskip

Following the previous notation, the system parameters are defined by the set
$\mathcal{S}=\{\mathbb{F}_q, [\beta_1, \beta_2, \beta_3], c_1, c_2, c_3\} $.

\subsection{The key agreement protocol}
The key agreement follows the well-known Diffie-Hellman paradigm. Any two
users $A, B$, willing to agree on a common value, which remains secret, set
up a system and agree on its parameters, as stated previously.

The protocol runs as follows:
\begin{enumerate}
\item User $A$ selects $n_A\in \mathbb{Z}_\ell$, with $\ell = q^2+q+1$,
computes
\[
[\gamma^A_1, \gamma^A_2, \gamma^A_3] =
                \oplus^{n_A} [\beta_1, \beta_2, \beta_3]\in\mathbb{F}_qP^2
\]
and sends it to user B.
\item User $B$ selects $n_B\in \mathbb{Z}_\ell $, computes
\[
[\gamma^B_1, \gamma^B_2, \gamma^B_3] =
                \oplus^{n_B} [\beta_1, \beta_2, \beta_3]\in\mathbb{F}_qP^2
\]
and sends it to user A.
\item User $A$ computes $k_A = \oplus^{n_A}[\gamma^B_1, \gamma^B_2, \gamma^B_3]$.
\item User $B$ computes $k_B = \oplus^{n_B}[\gamma^A_1, \gamma^A_2, \gamma^A_3]$.
\end{enumerate}
According to the definitions, the following equalities clearly hold:
\begin{eqnarray*}
k_A = \oplus^{n_A}[\gamma^B_1, \gamma^B_2, \gamma^B_3]
& = & \oplus^{n_A}\left(\oplus^{n_B}
[\beta_1, \beta_2,\beta_3]\right) \\
& = & \oplus^{n_B}\left(\oplus^{n_A} [\beta_1, \beta_2,\beta_3]\right) \\
& = & \oplus^{n_B}[\gamma^A_1, \gamma^A_2, \gamma^A_3] = k_B.
\end{eqnarray*}
Hence, the properties of the operation $\oplus$ in $\mathbb{G}$ ensure that actually
$k_A=k_B$, which is the common value expected as the output of the protocol.

\subsection{Cost of the $\oplus$ operation in $\mathbb{G}$}\label{cost}
Let $S$ and $P$ be the number of field operations in order to perform
an addition and a multiplication respectively in $\mathbb{F}_q$.
From the formulas \eqref{z's} it follows that the total number
of operations for computing $\mathbf{x}\oplus \mathbf{y}$ is equal to $10S+15P$,
once the $2S+3P$ precomputations of $c_1c_3$, $c_1c_2+c_3$, and
$(c_1)^2+c_2$ are assumed.

Additionally, two multiplications and one inversion are needed to eventually
project the resulting point back to $\mathbb{F}_qP^2$. However, in a typical
setting their cost can be neglected when compared with the relatively much
larger number of sums and products that are to be carried out.

\subsection{A toy example}
If we take the prime field $\mathbb{F}_p$, with $p=131$, it is case that
$p^2+p+1=17293$ is also prime. Accordingly, the group $\mathbb{G}$ is cyclic.
We set the parameters $c_1=13$, $c_2=18$, $c_3=73$, since the polynomial
$\chi (X) = X^3-13X^2-18X-73$ is irreducible in $\mathbb{F}_{131}$.

Let us take the projective point $X =  [126,16,1]$ as a generator of
$\mathbb{G}$. If we select now another projective point $Y=[86,120,1]$, we
find by exhaustive search the integer $n$ such that $Y = \oplus^{n}X$:
\begin{multline*}
[126,16,1] \rightarrow [117, 130, 1] \rightarrow [11, 15, 1]
\rightarrow [71, 56, 1] \\
\rightarrow [16, 98, 1] \rightarrow [72, 62, 1] \rightarrow [111, 125, 1]
\rightarrow [110, 130, 1] \\
\rightarrow [130, 114, 1] \rightarrow [86, 120, 1].
\end{multline*}

Since the operation has been iterated ten times, we conclude $Y=\oplus^{10}X$
for this particular pair, so that $\log_{X} Y = 10$.

\section{A more robust system}\label{S:Robust}

The security of the cryptosystem proposed in the previous sections can be
increased by extending the theory developed for a field to the case of
a unitary commutative ring $R$.

In fact, let $M$ be a free $R$-module of finite rank and let $A\colon
M\rightarrow M$ be an $R$-linear map with characteristic polynomial $\chi
_{A}(X)=\det(XI-\Lambda)$, $X$ being an indeterminate, $I$ the identity matrix
of order $r=\operatorname*{rank}M$, and $\Lambda$ the matrix of $A$ in an
arbitrary basis for $M$. According to \cite[III, \S 8, 11.Proposition 20]%
{Bourbaki} Cayley-Hamilton Theorem holds in this setting, namely $\chi_{A}(A)
=0$.

Hence, if $M=R^{3}$ and $\chi_{A}(X)=X^{3}-c_{1}X^{2}-c_{2}X-c_{3}$,
$c_{1},c_{2},c_{3}\in R$, then $A^{3}=c_{1}A^{2}+c_{2}A+c_{3}I$.

As above, we can define a degree-$3$ homogeneous polynomial in
$R[x_{1},x_{2},x_{3}]$ by setting $Q(x_{1},x_{2},x_{3})=\det\left(
x_{1}I+x_{2}\Lambda+x_{3}\Lambda^{2}\right)  $. As a computation shows, we
have%
\begin{align*}
Q(x_{1},x_{2},x_{3})  &  =-c_{2}x_{1}(x_{2})^{2}+\left[  (c_{2})^{2}%
-2(c_{1}c_{3})\right]  x_{1}(x_{3})^{2}+c_{1}(x_{1})^{2}x_{2}\\
&  +\left[  (c_{1})^{2}+2c_{2}\right]  (x_{1})^{2}x_{3}-(c_{2}c_{3}%
)x_{2}(x_{3})^{2}+(c_{1}c_{3})(x_{2})^{2}x_{3}\\
&  -\left(  c_{1}c_{2}+3c_{3}\right)  x_{1}x_{2}x_{3}+(x_{1})^{3}+c_{3}%
(x_{2})^{3}+(c_{3})^{2}(x_{3})^{3},
\end{align*}
thus proving that Lemma~\ref{lemma1} still holds in this case; i.e., $Q$
depends on $\chi_{A}$ only, but not on the matrix $\Lambda$.

The projective plane over $R$ is then defined as follows: $RP^{2}%
=(R^{3}\setminus\{\mathbf{0\}})/R^{\ast}$, where $R^{\ast}$ denotes the
multiplicative group of invertible elements in $R$ and $R^{\ast}$ acts on
$R^{3}\setminus\{\mathbf{0}\}$ by%
\[%
\begin{array}
[c]{llll}%
\lambda\cdot(x_{1},x_{2},x_{3})= & (\lambda x_{1},\lambda x_{2},\lambda
x_{3}), & \forall\lambda\in R^{\ast}, & \forall(x_{1},x_{2},x_{3})\in
R^{3}\setminus\{\mathbf{0}\}.
\end{array}
\]

Proceeding as in the previous sections, a composition law $\oplus\colon
R^{3}\times R^{3}\rightarrow R^{3}$, $(\mathbf{x},\mathbf{y})\mapsto
\mathbf{z}=\mathbf{x}\oplus\mathbf{y}$, $\mathbf{x}=(x_{1},x_{2},x_{3})$,
$\mathbf{y}=(y_{1},y_{2},y_{3})$, $\mathbf{z}=(z_{1},z_{2},z_{3})$, can be
defined by the formula
\begin{equation*}
\begin{array}
[c]{rl}%
z_{1}I+z_{2}A+z_{3}A^{2}= & \left(  x_{1}I+x_{2}A+x_{3}A^{2}\right)  \left(
y_{1}I+y_{2}A+y_{3}A^{2}\right)  ,
\end{array}
\end{equation*}
and similarly we deduce
\begin{equation}
\left(
\begin{array}
[c]{c}%
z_{1}\\
z_{2}\\
z_{3}%
\end{array}
\right)  =\left(
\begin{array}
[c]{ccc}%
x_{1} & c_{3}x_{3} & c_{1}c_{3}x_{3}+c_{3}x_{2}\\
x_{2} & x_{1}+c_{2}x_{3} & c_{2}x_{2}+c_{3}x_{3}+c_{1}c_{2}x_{3}\\
x_{3} & x_{2}+c_{1}x_{3} & x_{1}+(c_{1})^{2}x_{3}+c_{1}x_{2}+c_{2}x_{3}%
\end{array}
\right)  \left(
\begin{array}
[c]{c}%
y_{1}\\
y_{2}\\
y_{3}%
\end{array}
\right)  .\label{system}%
\end{equation}

The determinant of the matrix of (\ref{system}) is equal to $Q(x_{1}%
,x_{2},x_{3})$. Hence, $\oplus$ induces a composition law $\oplus\colon
Q^{-1}(R^{\ast})\times Q^{-1}(R^{\ast})\rightarrow Q^{-1}(R^{\ast})$. If $C$
denotes the set of classes modulo $R^{\ast}$ of points $\mathbf{x}\in R^{3}$
such that $Q(\mathbf{x})\in R\backslash R^{\ast}$, then $\oplus$ also induces
a composition law $\oplus\colon PQ^{-1}(R^{\ast})\times PQ^{-1}(R^{\ast
})\rightarrow PQ^{-1}(R^{\ast})$, where $PQ^{-1}(R^{\ast})=RP^{2}\setminus C$,
as if $Q(\mathbf{x})$ is invertible and $\lambda\in R^{\ast}$, then
$Q(\lambda\mathbf{x})=\lambda^{3}Q(\mathbf{x})$ is also invertible.

The same proof given in the case of a field shows that the composition law
$\oplus$ is associative, commutative and admits an identity element, which is
the vector $(1,0,0)$.

If $m=pq$ with $p\neq q$ prime integers, then from Chinese Remainder Theorem
there is a ring isomorphism between $\mathbb{Z}/m\mathbb{Z}$ and the product
ring $\mathbb{F}_{p}\times\mathbb{F}_{q}$. Hence each vector $\mathbf{x}\in
R^{3}$ can be assigned a pair $(\mathbf{x}^{\prime},\mathbf{x}^{\prime\prime
})$ in $(\mathbb{F}_{p})^{3}\times(\mathbb{F}_{q})^{3}$ and the group
$(\mathbb{Z}/m\mathbb{Z)}^{\ast}=(\mathbb{F}_{p})^{\ast}\times(\mathbb{F}%
_{q})^{\ast}$ acts on $R^{3}$ in the same way as $(\mathbb{F}_{p})^{\ast}$
acts on $(\mathbb{F}_{p})^{3}$ and $(\mathbb{F}_{q})^{\ast}$ does on
$(\mathbb{F}_{q})^{3}$.

Consequently, $\mathbf{x}\neq0$ if and only if at least one of its two
components $\mathbf{x}^{\prime},\mathbf{x}^{\prime\prime}$ is distinct from
$\mathbf{0}$, so that%
\begin{equation}%
\begin{array}
[c]{rl}%
R^{3}\setminus\{\mathbf{0\}}= & \left[  \{\mathbf{0}\}\times\left(
(\mathbb{F}_{q})^{3}\setminus\{\mathbf{0\}}\right)  \right]  \sqcup\left[
\left(  (\mathbb{F}_{p})^{3}\setminus\{\mathbf{0\}}\right)  \times
\{\mathbf{0}\}\right]  \sqcup\medskip\\
& \multicolumn{1}{r}{\left[  \left(  (\mathbb{F}_{p})^{3}\setminus
\{\mathbf{0\}}\right)  \times\left(  (\mathbb{F}_{q})^{3}\setminus
\{\mathbf{0\}}\right)  \right]  .}%
\end{array}
\label{descomposicion}%
\end{equation}
Therefore $(\mathbb{Z}/pq\mathbb{Z)}P^{2}=\mathbb{F}_{p}P^{2}\sqcup
\mathbb{F}_{q}P^{2}\sqcup\left(  \mathbb{F}_{p}P^{2}\times\mathbb{F}_{q}%
P^{2}\right)  $.

Moreover, letting $\mathbf{z}=(\mathbf{z}^{\prime},\mathbf{z}^{\prime\prime
})=\mathbf{x}\oplus\mathbf{y}$, as a computation shows, one obtains
$\mathbf{z}^{\prime}=\mathbf{x}^{\prime}\oplus\mathbf{y}^{\prime}$ and
$\mathbf{z}^{\prime\prime}=\mathbf{x}^{\prime\prime}\oplus\mathbf{y}%
^{\prime\prime}$, and $Q(\mathbf{x})$ is invertible if and only if
$Q(\mathbf{x})\operatorname{mod}p$ and $Q(\mathbf{x})\operatorname{mod}q$ both
are invertible in $\mathbb{Z}/p\mathbb{Z}$ and $\mathbb{Z}/q\mathbb{Z}$
respectively. If $\mathbf{x}\in R^{3}$ corresponds to $(\mathbf{x}^{\prime
},\mathbf{x}^{\prime\prime})$ in $(\mathbb{F}_{p})^{3}\times(\mathbb{F}%
_{q})^{3}$, then $Q(\mathbf{x})=(Q^{\prime}(\mathbf{x}^{\prime}),Q^{\prime
\prime}(\mathbf{x}^{\prime\prime}))$, where $Q^{\prime}(\mathbf{x}^{\prime
})=\det\left(  x_{1}^{\prime}I+x_{2}^{\prime}\Lambda^{\prime}+x_{3}^{\prime
}\Lambda^{\prime2}\right)  $, $Q^{\prime\prime}(\mathbf{x}^{\prime\prime
})=\det\left(  x_{1}^{\prime\prime}I+x_{2}^{\prime\prime}\Lambda^{\prime
\prime}+x_{3}^{\prime\prime}\Lambda^{\prime\prime2}\right)  $, and
$\Lambda^{\prime}=\Lambda\operatorname{mod}p$, $\Lambda^{\prime\prime}%
=\Lambda\operatorname{mod}q$. Hence%
\begin{equation}
Q^{-1}(R^{\ast})=\left\{  (\mathbf{x}^{\prime},\mathbf{x}^{\prime\prime}%
)\in(\mathbb{F}_{p})^{3}\times(\mathbb{F}_{q})^{3}:Q^{\prime}(\mathbf{x}%
^{\prime})\neq0,Q^{\prime\prime}(\mathbf{x}^{\prime\prime})\neq0\right\}
.\label{neq_neq}%
\end{equation}

We set%
\[
\left.
\begin{array}
[c]{ll}%
\chi^{\prime}(X)=X^{3}-c_{1}^{\prime}X^{2}-c_{2}^{\prime}X-c_{3}^{\prime}%
\in\mathbb{F}_{p}[X], & c_{i}^{\prime}=c_{i}\operatorname{mod}p\\
\chi^{\prime\prime}(X)=X^{3}-c_{1}^{\prime\prime}X^{2}-c_{2}^{\prime\prime
}X-c_{3}^{\prime\prime}\in\mathbb{F}_{q}[X], & c_{i}^{\prime\prime}%
=c_{i}\operatorname{mod}q
\end{array}
\right\}  1\leq i\leq3.
\]

If both $\chi^{\prime}$ and $\chi^{\prime\prime}$ are irreducible polynomials
in $\mathbb{F}_{p}[X]$ and $\mathbb{F}_{q}[X]$, respectively, then according
to Proposition~\ref{prop1}, the points of the associated curves $C^{\prime}$ and
$C^{\prime\prime}$ reduce to the origin; i.e., $Q^{\prime-1}(0)=\{\mathbf{0}%
_{p}\}$, $Q^{\prime\prime-1}(0)=\{\mathbf{0}_{q}\}$, where $\mathbf{0}_{p}$
and $\mathbf{0}_{q}$ denote the origin in $(\mathbb{F}_{p})^{3}$ and
$(\mathbb{F}_{q})^{3}$, respectively.

From (\ref{descomposicion}), taking (\ref{neq_neq}) into account, it follows:
$PQ^{-1}(R^{\ast})=\mathbb{F}_{p}P^{2}\times\mathbb{F}_{q}P^{2}$.
Consequently, we conclude that $PQ^{-1}(R^{\ast})\cong S_{p}\times S_{q}$,
where $S_{p}$ and $S_{q}$ are the subgroups given by%
\[
S_{p}=(\mathbb{F}_{p}P^{2}\times\{(1,0,0)\},\oplus),\quad S_{q}%
=(\{(1,0,0)\}\times\mathbb{F}_{q}P^{2},\oplus),
\]
and from Theorem~\ref{th4} we thus obtain

\begin{proposition}
If the polynomials $\chi^{\prime}$ and $\chi^{\prime\prime}$ are irreducible
in $\mathbb{F}_{p}[X]$ and $\mathbb{F}_{q}[X]$, respectively, then the group
$(PQ^{-1}(R^{\ast})=\mathbb{F}_{p}P^{2}\times\mathbb{F}_{q}P^{2},\oplus)$ is
isomorphic to the direct product of the cyclic groups $S_{p}$ and $S_{q}$.
Hence $(PQ^{-1}(R^{\ast}),\oplus)$ is cyclic if and only if $a=p^{2}+p+1$ and
$b=q^{2}+q+1$ are coprimes; i.e., $\gcd(a,b)=1$.
\end{proposition}

\begin{remark}
If $d=\gcd(a,b)$, then $a=da^{\prime}$, $b=db^{\prime}$, with $\gcd(a^{\prime
},b^{\prime})=1$. The cyclic subgroup $S$ in\ $\mathbb{Z}/a\mathbb{Z}%
\times\mathbb{Z}/b\mathbb{Z}$ spanned by $(1\operatorname{mod}%
a,1\operatorname{mod}b)$ is of order $\frac{ab}{d}$. As $d<pq$ and
$a=O(p^{2})$, $b=O(q^{2})$, it follows: $\frac{ab}{d}>\frac{O(p^{2}q^{2})}%
{pq}=O\left(  pq\right)  $, which indicates that in general the group $S$ is
large enough, even if $a$ and $b$ are not coprimes.
\end{remark}

\begin{remark}
It is clear that the group $(PQ^{-1}(R^{\ast}),\oplus)$ is also amenable a as
building block for a key-agreement protocol by choosing $R=\mathbb{Z}_m$, with
$m$ composite. Observe that its security is enhanced with respect to its
counterpart $\mathbb{F}_q$, $q$ a prime power, since the algorithms known to
be efficient to compute discrete algorithms only work in the multiplicative
group of a field. This means that one is forced to factorize $m$ in order to
apply such algorithms to the present case, thus increasing the time complexity
and the security of the system, though at the price of doubling the key length.
\end{remark}

\section{Conclusions}\label{S:Conclusions}

In this work, we have defined a group law, $\oplus$, over the
set $\mathbb{F} _qP^2$, and considered the discrete logarithm
problem associated to them.
We have analyzed their properties and stated the security of
the problem considered. Moreover, based on it, we have defined
a cryptographic key agreement protocol as one possible
application of this problem to public key cryptography. Finally, we shift the
system to the group $(PQ^{-1}(R^{\ast}),\oplus)$ over the ring
$\mathbb{Z}/pq\mathbb{Z}$, which turns out to be completely analogous to the
previous one and offers an enhanced security, though at the cost of some extra
key length.

As future work, we think that it is possible to extend
this discrete logarithm problem in order to define new
cryptographic protocols for encryption/decryption
and digital signatures, among others, in a similar way as
ElGamal or elliptic curve cryptosystems were defined
from the Diffie-Hellman key agreement protocol.

\medskip

\noindent \textit{Acknowledgments:\/}
This research has been partially supported by Ministerio de Economía,
Industria y Competitividad (MINECO), Agencia Estatal de Investigación
(AEI), and Fondo Europeo de Desarrollo Regional (FEDER, UE) under
project COPCIS, reference TIN2017-84844-C2-1-R, and by Comunidad de
Madrid (Spain) under project CYNAMON reference P2018/TCS-4566,
also co-funded by European Union FEDER funds.

\end{document}